\documentclass[runningheads,a4paper]{llncs}

\usepackage{amsmath}
\usepackage{amssymb,bbm}
\usepackage{graphicx}
\usepackage{psfrag}
\usepackage{lscape}

\newcommand{\R}{\mathbb{R}}
\newcommand{\Z}{\mathbb{Z}}
\newcommand{\IS}{\mathbb{S}}

\newcommand{\rel}{\R}

\def\ds{\displaystyle}

\newcommand{\fixedset}{\mathcal{J}}

\title{On a linear programming approach to the discrete Willmore boundary value problem and generalizations}
\titlerunning{A linear programming approach to the discrete Willmore boundary problem}

\author{Thomas Schoenemann\inst{1}\and Simon
Masnou\inst{2} \and Daniel Cremers\inst{3} }
\institute{Department of Mathematical Sciences, Lund University, Sweden\and Institut Camille Jordan, Universit\'e Claude-Bernard Lyon~1, CNRS, France \and Department of Computer Science, TU M\"unchen, Germany}
\begin{document}
\toctitle{Lecture Notes in Computer Science}
\tocauthor{Authors' Instructions}

\maketitle

\begin{abstract}
We consider the problem of finding (possibly non connected) discrete surfaces spanning a finite set of
discrete boundary curves in the three-dimensional space and minimizing (globally) a discrete energy involving mean curvature. Although we consider a fairly general class of energies, our main focus is on the Willmore energy, i.e. the total squared mean curvature.

Most works in the literature have been devoted to the approximation of a surface evolving by the Willmore flow and, in particular,  to the approximation of the so-called Willmore surfaces, i.e., the critical points of the Willmore energy.
Our purpose is to address the delicate task of approximating {\it global} minimizers of the energy under boundary constraints. 

The main contribution of this work is to translate the nonlinear boundary value problem into an integer linear program, using a natural formulation involving pairs of elementary triangles chosen in a pre-specified dictionary and allowing self-intersection.

The reason for such strategy is the well-known existence of algorithms that can compute {\it global minimizers} of a large class of linear optimization problems, however at a significant computational and memory cost. 
The case of integer linear programming is particularly delicate and usual strategies consist in relaxing the integral constraint $x\in\{0,1\}$ into $x\in[0,1]$ which is easier to handle. Our work focuses essentially on the connection between the integer linear program and its relaxation. We prove that:
\begin{itemize}
\item One cannot guarantee the total unimodularity of the constraint matrix, which is a sufficient condition for the global solution of the relaxed linear program to be always integral, and therefore to be a solution of the integer program as well;
\item Furthermore, there are actually experimental evidences that, in some cases, solving the relaxed problem yields a fractional solution.
\end{itemize}
These facts prove that the problem cannot be tackled with classical linear programming solvers, but only with pure integer linear solvers. Nevertheless, due to the very specific structure of the constraint matrix here, we strongly believe that it should be possible in the future to design ad-hoc integer solvers that  yield high-definition approximations to solutions of several boundary value problems involving mean curvature, in particular the Willmore boundary value problem.


\end{abstract}

\section{Introduction}
The Willmore energy of an immersed compact oriented surface $f:\Sigma\to\R^N$ with boundary $\partial\Sigma$ is defined as
$${\cal W}(f)=\int_\Sigma|H|^2dA+\int_{\partial\Sigma}\kappa\,ds$$
where $H$ is the mean curvature vector on $\Sigma$, $\kappa$ the
geodesic curvature on $\partial\Sigma$, and $dA$, $ds$ the induced
area and length metrics on $\Sigma$, $\partial\Sigma$. The Willmore
energy of surfaces with or without boundary plays an important role in
geometry, elastic membranes theory, strings theory, and image
processing. Among the many concrete optimization problems where the
Willmore functional appears, let us mention for instance the modeling
of biological membranes, the design of glasses, and the smoothing of
meshed surfaces in computer graphics. The Willmore energy is the
subject of a long-standing research not only due to its relevance to
some physical situations but also due to its fundamental property of
being conformal invariant, which makes it an interesting substitute to
the area functional in conformal geometry. Critical points of ${\cal
  W}$ with respect to interior variations are called Willmore
surfaces. They are solutions of the Euler-Lagrange equation
$\delta{\cal W}=0$ whose expression is particularly simple when $N=3$:
$\Delta H+2H(H^2-K)=0$, being $K$ the Gauss curvature.  It is known
since Blaschke and Thomsen~\cite{pinkall-sterling} that stereographic
projections of compact minimal surfaces in $\IS^3\subset\R^4$ are always
Willmore surfaces in $\R^3$. However, Pinkall exhibited
in~\cite{Pinkall} an infinite series of compact embedded Willmore
surfaces that are not stereographic projections of compact embedded
minimal surfaces in $\IS^3$. Yet Kusner conjectured~\cite{Kusner} that
stereographic projections of Lawson's $g$-holed tori in $\IS^3$
should be global minimizers of ${\cal W}$ among all genus $g$
surfaces. This conjecture is still open, except of course for the case $g=0$ where the round sphere is known to be the unique global minimizer.

The existence of smooth surfaces that minimize the Willmore energy spanning a given boundary and a conormal field has been proved by Sch{\"a}tzle in~\cite{Sch}. Following the notations in~\cite{Sch}, we consider a smooth embedded closed oriented curve $\Gamma\subset\R^N$ together with a smooth unit normal field $n_\Gamma\in N_\Gamma$ and we denote as $\pm\Gamma$ and $\pm n_\Gamma$ their possible orientations. We assume that there exist oriented extensions of $\pm \Gamma$, $\pm n_\Gamma$, that is, there are compact oriented surfaces $\Sigma_-,\,\Sigma_+\subset\R^N$ with boundary $\partial\Sigma_\pm=\pm\Gamma$ and conormal vector field $\operatorname{co}_{\Sigma_\pm}=\pm n_\Gamma$ on $\partial\Sigma_\pm$. We also assume that there exists a bounded open set $B\supset\Gamma$ such that the set
\begin{multline*}\{\Sigma_\pm\mbox{ oriented extensions of } (\Gamma, n_\Gamma),\; \Sigma_+\mbox{ connected },\\
\Sigma_+\cup\Sigma_-\subset B,\;{\cal W}(\Sigma_+\cup\Sigma_-)< 8\pi\mbox\}\end{multline*}
is not empty. The condition on energy ensures that $\Sigma_+\cup\Sigma_-$ is an embedding.

It follows from~\cite{Sch}, Corollary 1.2, that the Willmore boundary problem associated with $(\Gamma,n_\Gamma)$ in $B$ has a solution, i.e., there exists a compact, oriented, connected, smooth surface $\Sigma\subset B$ with $\partial\Sigma=\Gamma$, $\operatorname{co}_\Sigma=n_\Gamma$ on $\partial\Sigma$, and 
$$W(\Sigma)=\min\{W(\tilde\Sigma),\,\tilde\Sigma\mbox{ smooth},\;\tilde\Sigma\subset B,\;\partial\tilde\Sigma=\Gamma,\;\operatorname{co}_{\tilde\Sigma}=n_\Gamma\mbox{ on }\partial\tilde\Sigma\}$$

There have been many contributions to the numerical simulation of Willmore surfaces in space dimension $N=3$. Among them, Hsu, Kusner and Sullivan have tested experimentally in~\cite{HsuKusnerSullivan92} the validity of Kusner's conjecture: starting from a triangulated polyhedron in $\R^3$ that is close to a Lawson's surface of genus $g$, they let it evolve by a discrete Willmore flow using Brakke's Surface Evolver~\cite{Brakke-92} and check that the solution obtained after convergence is ${\cal W}$-stable. Recent updates that Brakke brought to its program give now the possibility to test the flow with various discrete definitions of the mean curvature. Mayer and Simonett~\cite{MayerSim} introduce a finite difference scheme to approximate axisymmetric solutions of the Willmore flow.  Rusu~\cite{Rusu} and Clarenz et al.~\cite{Clarenz-et-al-04} use a finite elements approximation of the flow to compute the evolution of surfaces with or without boundary. In both works, position and mean curvature vector are taken as independent variables, which is also the case of the contribution by Verdera et al.~\cite{Verdera03}, where a triangulated surface with a hole in it is restored using the following approach: by the coarea formula, the Willmore energy (actually a generalization to other curvature exponents) is replaced with the energy of an implicit and smooth representation of the surface, and the mean curvature term is replaced by the divergence of an unknown field that aims to represent the normal field. Droske and Rumpf~\cite{DroskeRumpf04} propose a finite element approach to the Willmore flow but replace the standard flow equation by its level set formulation. The contribution of Dziuk~\cite{dziuk} is twofold: it provides a finite element approximation to the Willmore flow with or without boundary conditions that can handle as well embedded or immersed surfaces (turning the surface problem into a quasi-planar problem), and a consistency result showing the convergence of both the discrete surface and the discrete Willmore energy to the continuous surface and its energy when the approximated surface has enough regularity. Bobenko and Schr{\"o}der~\cite{BobenkoSchroeder05} use a difference strategy: they introduce a discrete notion of mean curvature for triangulated surfaces computed from the circles circumscribed to each triangle that shares with the continuous definition a few properties, in particular the invariance with respect to the full M{\"o}bius group in $\R^3$. This discrete definition is vertex-based and a discrete flow can be derived. Based also on several axiomatic constraints but using a finite elements framework, Wardetzky et al.~\cite{Wardetzky-et-al-07} introduce an edge-based discrete Willmore energy for triangulated surfaces. Olischl\"{a}ger and Rumpf~\cite{Olischlager:2009} introduce a two step time discretization of the Willmore flow that extends to the Willmore case, at least formally, the discrete time approximation of the mean curvature motion due to Almgren, Taylor, and Wang~\cite{Almgren-Taylor-Wang-1993}, and Luckhaus and Sturzenhecker~\cite{Luckhaus-Styrzenhecker-1995}. The strategy consists in using the mean curvature flow to compute an approximation of the mean curvature and plug it in a time discrete approximation of the Willmore flow. Grzibovskis and Heintz~\cite{Heintz03}, and Esedoglu et al.~\cite{EsedogluRuuthTsai-06} discuss how 4th order flows can be approximated by iterative convolution with suitable kernels and thresholding.

While all the previous approaches yield approximations of critical points of the Willmore energy, our motivation in this paper is to approximate global minimizers of the energy. This is an obviously nontrivial task due to the high nonlinearity and nonconvexity of the energy. Yet, for the simpler area functional, Sullivan~\cite{Sullivan-94} has shown with a calibration argument that the task of finding minimal surfaces can be turned into a linear problem. Even more, when a discrete solution is seeked among surfaces that are union of faces in a cubic grid partition of $\R^3$, he proved that the minimization of the linear program is equivalent to solving a minimum-cost circulation network flow problem, for which efficient codes have been developed by Boykov and Kolmogorov~\cite{Boykov-Kolmogorov-01} after Ford and Fulkerson~\cite{Ford-Fulkerson-62}. Sullivan~\cite{Sullivan-94} did not provide experiments in his paper but this was done recently by Grady~\cite{Grady-09}, with applications to the segmentation of medical images. 

The linear formulation that we propose here is based on two key ideas: the concept of surface continuation constraints that has been pioneered by Sullivan~\cite{Sullivan-94} and Grady~\cite{Grady-09}, and the representation of a triangular surface using pairs of triangles. With this representation and a suitable definition of discrete mean curvature, we are able to turn into a linear formulation the task of minimizing discrete representations of any functional of the form
$$W_\varphi(\Sigma)=\int_\Sigma \varphi(x,n,H)dA$$
among discrete immersed surfaces with boundary constraints:
$$\partial\Sigma=\Gamma,\quad\operatorname{co}_{\tilde\Sigma}=n_\Gamma\mbox{ on }\partial\Sigma.$$
In the expression of $W_\varphi(\Sigma)$, $x$ denotes the space variable, $n$ the normal vector field on $\Sigma$ and $H$ the mean curvature vector. The linear problem we obtain involves integer-valued unknowns and does not seem to admit any simple graph-based equivalent. We will therefore discuss whether classical strategies for linear optimization can be used.

The paper is organized as follows: in section~\ref{sec:1} we discuss both the chosen representation of surfaces and the definition of discrete mean curvature. In section~\ref{sec:2} we present a first possible approach yielding a quadratic energy. We present in section~\ref{sec:3} our linear formulation and discuss whether it can be tackled by classical linear optimization techniques.

\section{Discrete framework}\label{sec:1}
\subsection{Triangular meshes from a set of pre-defined triangles}\label{sec:mesh}
The equivalence shown by Sullivan between finding minimal surfaces and solving a flow problem holds true for discrete surfaces defined as a connected set of cell faces in a cellular complex discrete representation of the space. We will consider here polyhedral surfaces defined as union of triangles with vertices in (a finite subset of) the cubic lattice $\epsilon\Z^3$ where $\epsilon=\frac 1 n$ is the resolution scale. Not all possible triangles are allowed but only those respecting a specified limit on the maximal edge length.
We assume that each triangle, as well as each triangle edge, is represented twice, once for each orientation. We let ${\cal I}$ denote the collection of oriented triangles, $N=|{\cal I}|$ its cardinality, and $M$ the number of oriented triangle edges. The constrained boundary is given as a contiguous oriented set of triangle edges. The orientation of the boundary constrains the spanning surfaces since we will allow only spanning triangles whose orientation is compatible.

In this framework, one can represent a triangular mesh as
a binary indicator vector $x = \{0,1\}^N$ where $1$ means that the
respective triangle is present in the mesh, $0$ that it is not. Obviously, not all binary indicator vectors can be associated with a triangular surface since the corresponding triangles may not be contiguous. However, as discussed by Grady~\cite{Grady-09} and, in a slightly different setting, by Sullivan~\cite{Sullivan-thesis,Sullivan-94}, it is possible to write in a linear form the constraint that only binary vectors that correspond to surfaces spanning the given boundary are considered. We will see that using the same approach here turns the initial boundary value problem into a quadratic program. Another formulation will be necessary to get a linear problem.

\subsection{Admissible indicator vectors: a first attempt}
\begin{figure}
\begin{center}
\psfrag{e1}{$e_1$}
\psfrag{e2}{$e_2$}
\psfrag{e3}{$e_3$}
\psfrag{e4}{$e_4$}
\includegraphics[width=0.25 \textwidth]{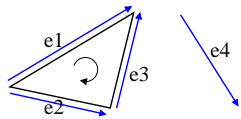}
\end{center}
\caption{Incidence of oriented triangles and edges. $e_1$ is positively incident to the oriented triangle, $e_2$ and $e_3$ are negatively incident, and $e_4$ is not incident to the triangle.}
\label{fig:incidence}
\end{figure}

To define the set of admissible indicator vectors, we first consider a
relationship between oriented triangles and oriented edges which is called
\emph{incidence}: a triangle is positive incident to an edge if the
edge is one of its borders and the two agree in orientation. It is
negative incident if the edge is one of its borders, but in the
opposite orientation. Otherwise it is not incident to the edge. For
example, the triangle in Figure \ref{fig:incidence} is positive
incident to the edge $e_1$, negative incident to $e_2$ and $e_3$ and
not incident to $e_4$.

Being defined as above the set of $N$ oriented triangles and their $M$ oriented edges, we introduce the matrix $B=(b_{ij})_{\begin{subarray}{l}
i\in\{1,\cdots,N\}\\
j\in\{1,\cdots,M\}\end{subarray}}$ whose element $b_{ij}$ gives account of the incidence between triangle $i$ and edge $j$. More precisely

$$
b_{ij} = 
\begin{cases}
1 & \mbox{if edge $i$ is an edge of triangle $j$ with same orientation}\\
-1 & \mbox{if edge $i$ is an edge of triangle $j$ with opposite orientation}\\
0 & \mbox{otherwise}
\end{cases}
$$
The knowledge of which edges are present in the set of prescribed boundary
segments is expressed as a vector $r \in \{-1,0,1\}^{M}$ with
$$
r_j = 
\begin{cases}
1 & \mbox{if the oriented boundary contains the edge $j$}
\\&\mbox{\qquad with agreeing orientation}\\
-1 & \mbox{if the oriented boundary contains the edge $-j$}\\
&\mbox{\qquad with opposing orientation}\\
0 & \mbox{otherwise}
\end{cases}
$$

With these notations set up we can now describe the equation system
defining that a vector $x \in \{0,1\}^N$ encodes an oriented triangular mesh
with the pre-specified oriented boundary. This system has one equation for each
edge. If the edge is not contained in the given boundary, this
equation expresses that, among all triangles indicated by $x$ that contain the edge, there are as many triangles with same orientation as the edge as triangles with opposite orientation. If the edge is contained
in the boundary with coherent orientation, there must be one more
positive incident triangle than negative incident. If it is contained
with opposite orientation, there is one less positive than negative
incident. Altogether the constraint for edge $j$ can be expressed
as the linear equation
$$
\sum\limits_{i} b_{ij}\, x_i = r_j 
$$ 
and the entire system as
\begin{equation}
B\, x = r.
\end{equation}
So far, we did not incorporate the conormal constraint. Actually not all conormal constraints are possible, exactly like not all discrete curves can be spanned in our framework but only union of edges of dictionary triangles, i.e. the collection of triangles defined in the previous section that determine the possible surfaces. For the conormal constraint, only the conormal vectors that are tangent to dictionary triangles sharing an edge with the boundary curve are allowed. Then the conormal constraint can be easily plugged into our formulation by simply imposing the corresponding triangles to be part of the surface, see Figure~\ref{fig:tribound}, and by defining accordingly a new boundary indicator vector $\tilde r$.
\begin{figure}
\begin{center}
\includegraphics[width=0.45 \textwidth]{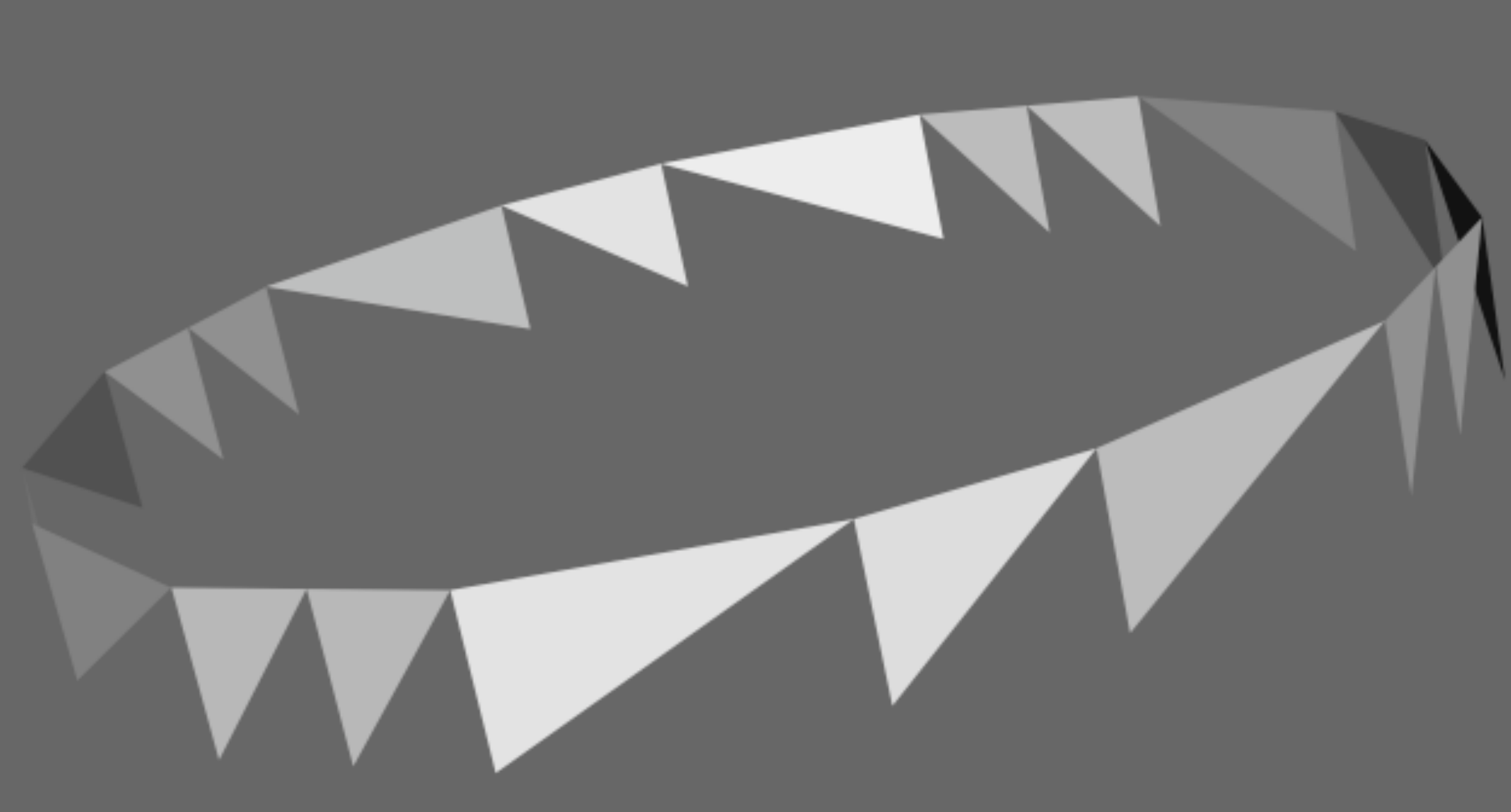}
\caption{The boundary and conormal constraints can be imposed by pre-specifying suitable triangles to be part of the surface.}
\label{fig:tribound}
\end{center}
\end{figure}

Denoting as ${\cal J}$ the collection of those additional triangles, the complete constraint reads
\begin{equation}\label{eq:const}
\left\{\begin{array}{l} 
B\, x = \tilde r\\
x_j=1,\;j\in {\cal J}
\end{array}\right.
\end{equation}
We discuss in the next section how discrete mean curvature can be evaluated in this framework.

\subsection{Discrete mean curvature on triangular meshes}
The various definitions of discrete mean curvature that have been proposed in the literature obviously depend on the chosen discrete representations of surfaces. Presenting and discussing all possible definitions is out of the scope of the present paper. The important thing to know is that there is no fully consistent definition: the pointwise convergence of mean curvature cannot be guaranteed in general but only in specific situations~\cite{Hildebrandt05,Morvan-2008}. Among the many possible definitions, we will use the edge-based one proposed by Polthier~\cite{Polthier-Hab} for it suits with our framework. Recalling that, in the smooth case but also for generalized surfaces like varifolds~\cite{Simon-83}, the first variation of the area can be written in terms of the mean curvature, the definition due to Polthier of the mean curvature vector at an interior edge $e$ of a simplicial surface reads
\begin{equation}
H(e)=|e|\cos\frac{\theta_e}2N_e\label{eq:meancurv}
\end{equation}
where $|e|$ is the edge-length, $\theta_e$ is the dihedral angle between the two triangles adjacent to $e$, and $N_e$ is the angle bisecting unit normal vector, i.e., the unit vector collinear to the half sum of the two unit vectors normal to the adjacent triangles (see figure~\ref{fig:local-triangle}). Remark that this formula is a discrete counterpart of the continuous $H=\kappa_1+\kappa_2$ depending on the principal curvatures, which is used in many papers for simplicity as definition of mean curvature. When the correct continuous definition $H=\frac 1 2(\kappa_1+\kappa_2)$ is used, the formulas above and hereafter should be adapted.
\begin{figure}
\begin{center}
\includegraphics[width=0.35 \textwidth]{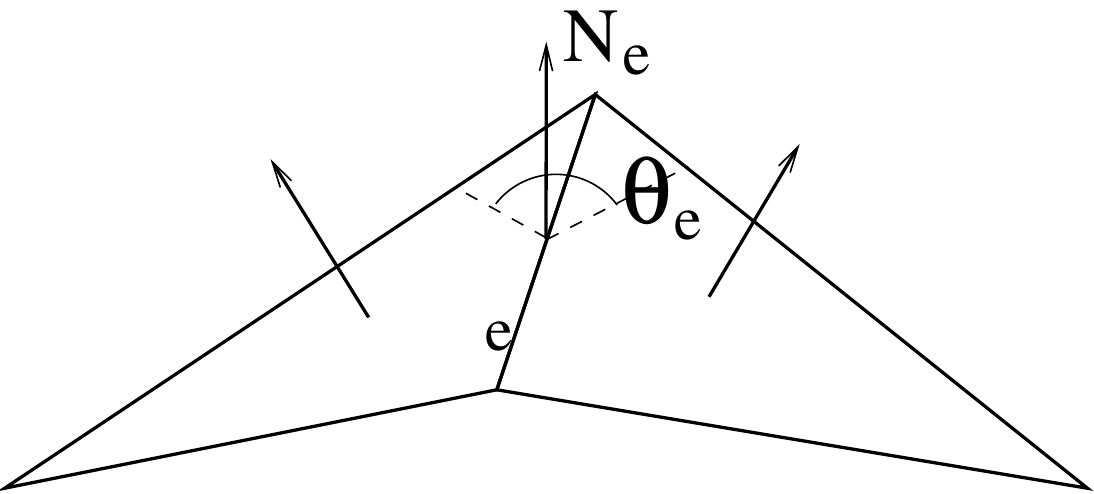}
\caption{The edge-based definition of a discrete mean curvature vector due to Polthier~\cite{Polthier-Hab} depends on the dihedral angle $\theta_e$ and the angle bisecting unit normal vector $N_e$.}
\label{fig:local-triangle}
\end{center}
\end{figure}
The justification of formula~\eqref{eq:meancurv} by Polthier~\cite{Polthier-Hab,Polthier-Clay} is as follows: it is exactly the gradient at any point $m\in e$ of the area of the two triangles $T_1$ and $T_2$ adjacent to $e$, and this gradient does not depend on the exact position of $m$. Indeed, one can subdivide $T_1$, $T_2$ in four triangles $T'_i$, $i\in\{1,\cdots,4\}$ having $m\in e$ as a vertex and such that $T_1=T'_1\cup T'_2$ and $T_2=T'_3\cup T'_4$. The area of each triangle is half the product of the opposite edge's length and the height. Therefore, if $e_i$ is the positively oriented edge opposite to $m$ in the triangle $T'_i$ and $J_1$, $J_2$ the rotations in the planes of $T_1$, $T_2$ by $\frac\pi 2$, the area gradients of $T'_i$, $i\in\{1,\cdots, 4\}$ at $m$ are $\frac 1 2 J_1e_1$, $\frac 1 2 J_1e_2$, $\frac 1 2 J_2e_3$, $\frac 1 2 J_2e_4$. The sum is the total area gradient of $T_1\cup T_2$ at $m$ and equals $\frac 1 2 (J_1 e+J_2 e)$, which coincides with the formula above.

As discussed by Wardetsky et al. using the Galerkin theory of approximation, this discrete mean curvature is an integrated quantity: it scales as $\lambda$ when each space dimension is rescaled by $\lambda$. A pointwise discrete mean curvature rescaling as $\frac 1{\lambda}$ is given by (see~\cite{Wardetzky-et-al-07})
$$H^{{\rm pw}}(e)=\frac{3|e|}{A_e}\cos\frac{\theta_e}2N_e,$$
where $A_e$ denotes the total area of the two triangles adjacent to $e$. The factor $3$ comes from the fact that, when the mean curvatures are summed up over all edges, the area of each triangle is counted three times, once for each edge. Then a discrete counterpart of the energy $\ds\int_\Sigma \varphi(H)\, dA$ is given by
\begin{equation}
\sum_{{\rm edges}\; e}\frac{A_{e}}{3} \varphi(\frac{3|e|}{A_{e}}\cos\frac{\theta_{e}}2N_{e}).\label{eq:globenerg}
\end {equation}

\noindent In particular, the edge-based total squared mean curvature is
\begin{equation}\label{eq:willm}
\sum_{{\rm edges}\; e}\frac{3|e|^2}{A_{e}}(\cos\frac{\theta_e}2)^2.
\end{equation}


\section{A quadratic program for the minimization of the discrete Willmore energy}\label{sec:2}
Ultimately we are aiming at casting the optimization problem in a form
that can be handled by standard linear optimization software. Having in mind the framework described above where a discrete surface spanning the prescribed discrete boundary is given as a collection of oriented triangles satisfying equation~\eqref{eq:const} and chosen among a pre-specified collection of triangles, a somewhat natural direction at first glance seems to be solving a \emph{quadratic program}. Like in section~\ref{sec:mesh}, let us indeed denote as $(x_i)$ the collection of binary variables associated to the ``dictionary'' of triangles $(T_i)$ and define
\begin{itemize}
\item $e_{ij}$ the common edge to two adjacent triangles $T_i$ and $T_j$;
\item $\theta_{ij}$ the corresponding dihedral angle;
\item $N_{ij}$ the angle bisecting unit normal;
\item $A_{ij}$ the total area of both triangles.
\end{itemize}
Then a continuous energy of the form $\ds\int_\Sigma \varphi(x,n,H)dA$ can be discretized as
\begin{equation}\label{eq:formgen}
\sum_{i,j} q_{ij}\, x_i\,
x_j\end{equation}
with $\ds\quad q_{ij}=\left\{\begin{array}{ll}
\ds\frac 1 2\frac{A_{ij}}{3} \varphi(e_{ij},N_{ij},\frac{3|e_{ij}|}{A_{ij}}\cos\frac{\theta_{ij}}2N_{ij})&\mbox{if $i\not=j$ are adjacent}\\[1mm]
\tilde\varphi(T_i,N_i)&\mbox{if $i=j$}\\[1mm]
0&\mbox{otherwise}\end{array}\right.$
\par
\noindent where $\tilde\varphi$ allows to incorporate dependences on each triangle $T_i$'s position and unit normal $N_i$.
In particular, the discrete Willmore energy is
\begin{equation}\sum_{i,j}q^w_{ij}x_i\,
x_j\label{eq:formwill}\end{equation}
with
$$q^w_{ij}=\left\{\begin{array}{ll}
\ds\frac{3|e_{ij}|^2}{2A_{ij}}(\cos\frac{\theta_{ij}}2)^2&\mbox{if $i\not=j$ are adjacent}\\[1mm]
0&\mbox{otherwise}\end{array}\right.$$

Assuming that the maps $\varphi$ and $\tilde\varphi$ are positive-valued, both energy matrices $Q=(q_{ij})$ and $Q^w=(q^w_{ij})$ are symmetric matrices in ${\rel^+}^{N \times
N}$, and the minimization of either~\eqref{eq:formgen} or~\eqref{eq:formwill} with boundary constraints turns to be the following quadratic program with linear and integrality constraints:
\begin{eqnarray*}
& \min\limits_x & \langle Q\, x,x\rangle \\
& \mbox{such that} & B\, x = r\\
&&  x_i = 1\ \, \forall i \in \fixedset\\
& & x \in \{0,1\}^N \quad \quad .
\end{eqnarray*}
We know of no solution to solve this problem efficiently due to the integrality
constraint. What is worse, even the relaxed problem where
one optimizes over $x \in [0,1]^N$ is very hard to solve: terms of the
form $x_i x_j$ with $i\neq j$ are indefinite, so (unless $Q$ has a
dominant diagonal) the objective function is a non-convex one.

Moreover, a solution to the relaxed problem would not be of practical
use: already for the 2D-problem of optimizing curvature energies over
curves in the plane, the respective quadratic program favors
fractional solutions. The relaxation would therefore not be useful for
solving the integer program. However, in this case Amini et
al. \cite{Amini-et-al-90} showed that one can solve a linear program
instead. This inspired us for the major contribution of this work: to
cast the problem as an integer linear program.

\section{An integer linear programming approach}\label{sec:3}
\subsection{Augmented indicator vectors}

The key idea of the proposed integer linear program is to consider
additional indicator vectors. Aside from the indicator variables $x_i$
for basic triangles, one now also considers entries $x_{ij}$ corresponding to
\emph{pairs} of adjacent triangles. Such a pair is called \emph{quadrangle} in the following. We will denote $\Hat x$ the augmented vector $(x_1,\cdots,x_N,\cdots, x_{ij},\cdots)$ where $i\not=j$ run over all indices of adjacent triangles. The cost function can be easily written in a linear form with this augmented vector, i.e. it reads
$$\sum w_k\Hat x_k$$
with (see the notations of the previous section)
$$w_k=\left\{\begin{array}{ll}
q_{ii}&\mbox{if }\Hat x_k=x_i\\
q_{ij}&\mbox{if }\Hat x_k=x_{ij}\end{array}\right.$$

The major problem to overcome is how to set up a system of constraints
that guarantees consistency of the augmented vector: the indicator
variable $x_{ij}$ for the pair of triangles $i$ and $j$ should be~$1$
if and only if both the variables $x_i$ and $x_j$ are $1$. Otherwise
it should be~$0$. In addition, one again wants to optimize only over
indicator vectors that correspond to a triangular mesh.

To encode this in a linear constraint system, a couple of changes are
necessary. First of all, we will now have a constraint for each pair
of triangle and adjacent edge. Secondly, edges are no longer
oriented. Still, the set of pre-specified indices $\fixedset$ implies
that the orientation of the border is fixed - we still require that
for each edge of the boundary an adjacent (oriented) triangle is
fixed to constrain the conormal information.

To encode the constraint system we introduce a modified notion of
incidence. We are no longer interested in incidence of triangles and
edges. Instead we now consider the incidence of both triangles and
quadrangles to pairs of triangles and (adjacent) edges.


For convenience, we define that triangles are positive incident to
a pair of edge and triangle, whereas all quadrangles are negative incident. 

We propose an incidence matrix where lines correspond to pairs (triangle, edge) and columns to either triangles or quadrangles. The entries of this incidence matrix are either the incidence of a pair (triangle, edge) with a triangle, defined as
$$ 
d((\mbox{triangle }k,\mbox{edge } e),\mbox{triangle } i) = \begin{cases}
1 & \mbox{if } i=k,\ e \mbox{ is an edge of triangle } i\\
0 & \mbox{otherwise}
\end{cases},
$$
or the incidence of a  pair (triangle, edge) with a quadrangle, defined as
$$
d((\mbox{triangle }k,\mbox{edge } e),\mbox{quadrangle  }ij) =
\begin{cases}
-1 & \mbox{if }  i\!=\!k\mbox{ or } j\!=\!k \mbox{ and }  
i, j \mbox{ share } e
 \\
\ 0 & \mbox{otherwise}
\end{cases}.
$$
The columns of this incidence matrix are of two types: either with only 0's and exactly three $1$ (a column corresponding to a triangle $T$, whose three edges are found at lines $(T,e_1)$, $(T,e_2)$, $(T,e_3)$), or with only 0's and exactly two $(-1)$'s (a column corresponding to a quadrangle $(T_1,T_2)$ that matches with lines $(T_1,e_{12})$ and $(T_2,e_{12})$).

Again, both the conormal constraints and the boundary edges can be imposed by imposing additional triangles indexed by a collection $\fixedset$ of indices. The general constraint has the form
$$
\sum_i d((x_k,e),x_i)\, +\, \sum_{i,j} d((x_k,e),x_{ij}) = r'_{(k,e)},
$$ where the right-hand side depends whether the edge $e$ is shared by two triangles of the surface (and even several quadrangles in case of self-intersection), or belongs to the new boundary indicated by the additional triangles. If $e$ is an inner edge, then the sum must be zero due to our definition of $d$, otherwise there is an adjacent
triangle, but no adjacent quadrangle, so the right-hand side should be~$1$:
$$ r'_{(k,e)} = \begin{cases} 1 & \mbox{if } k \in \fixedset , e
\mbox{ is part of the modified boundary}\\ 0 & \mbox{otherwise}
\end{cases}
$$
To sum up, we get the following integer linear program:
\begin{eqnarray}
& \min\limits_{\Hat x} & \langle w,\Hat x\rangle \label{eq:ilp}\\
&\mbox{such that} & D\,\Hat x = r'  \nonumber \\
&& \Hat x_j = 1\quad \forall j \in \fixedset \nonumber \\
&& \Hat x_i \in \{0,1\}\quad \forall i \in \{1,\ldots, \Hat N\} \nonumber 
\end{eqnarray}
where $\Hat N$ is the total number of entries in $\Hat x$, namely all triangles plus all pairs of adjacent triangles. It is worth noticing that such formulation allows triangle surfaces with self-intersection. 

\subsection{On the linear programming relaxation}

Solving integer linear programs is an NP-complete problem, see e.g.
\cite[chapter 18.1]{Schrijver-book}. This implies that, to the noticeable exception of a few particular problems~\cite{Schrijver-book}, no
efficient solutions are known. As a consequence one often resorts to
solving the corresponding linear programming (LP) relaxation, i.e. one
drops the integrality constraints. In our case this means to solve the
problem:
\begin{eqnarray}
& \min\limits_{\Hat x} & \langle w, \Hat x\rangle \label{eq:lpf}\\
&\mbox{such that} & D\, \Hat x = r'  \nonumber \\
&& \Hat x_j = 1\quad \forall j \in \fixedset \nonumber \\
&& 0 \le \Hat x_i \le 1\quad \forall i \in \{1,\ldots, \Hat N\} \nonumber 
\end{eqnarray}
or, equivalently, by suitably augmenting $D$ and $r'$ in order to incorporate the second constraint $\Hat x_j = 1$, $\forall j \in \fixedset $:
\begin{eqnarray}
& \min\limits_{\Hat x} & \langle w, \Hat x\rangle \label{eq:lp}\\
&\mbox{such that} & \Hat D \Hat x = \Hat r  \nonumber \\
&& 0 \le \Hat x_i \le 1\quad \forall i \in \{1,\ldots, \Hat N\} \nonumber 
\end{eqnarray}
There are various algorithms for solving this problem, the most classical being the simplex algorithm and several interior point algorithms. Let us now discuss the conditions under which these relaxed solutions are also solutions of the original integer linear program. Recalling the basics of LP-relaxation~\cite{Schrijver-book}, the set of admissible solutions
$$P=\{\Hat x\in \R^{\Hat N},\;\Hat D \Hat x= \Hat r,\; 0\leq x\leq 1\}$$ is a polyhedron, i.e. a finite intersection of half-spaces in $\R^{\Hat N}$. A classical result states that minimizing solutions for the linear objective functions can be seeked among the extremal points of $P$ only, i.e. its vertices. Denoting $P_e$ the integral envelope of $P$, that is the convex envelope of $P\cap\Z^{\Hat N}$, another classical result states that $P$ has integral vertices only (i.e. vertices with integral coordinates) if and only if $P=P_e$

\par
Since $P=\{\Hat x\in \R^{\Hat N},\,\Hat D \Hat x= \Hat r, 0\leq \Hat x\leq 1\}$, according to Theorem 19.3 in~\cite{Schrijver-book}, a {\it sufficient} condition for having $P=P_e$ is the property of $B$ being totally unimodular, i.e. any square submatrix has determinant either $0$, $-1$ or~$1$. Under this condition, any extremal point of $P$ that is a solution of 
$$\min_{\Hat D \Hat x= \Hat r,\,\Hat x_i\in[0,1]}\langle w,\Hat x\rangle$$
has integral coordinates therefore is a solution of the original integer linear program
$$\min_{\Hat D \Hat x= \Hat r,\,\Hat x_i\in\{0,1\}}\langle w,\Hat x\rangle.$$
Theorem 19.3 in~\cite{Schrijver-book} mentions an interesting characterization of total unimodularity due to Paul Camion~\cite{camion-65}: a matrix is totally unimodular if, and only if, the sum of the entries of every Eulerian square submatrix (i.e. with even rows and columns) is divisible by four.

Unfortunately, we can prove that, as soon as the triangle space is rich enough, the incidence matrix $\Hat D$ does not satisfy Camion's criterion, therefore is not totally unimodular, and neither are the matrices for richer triangles spaces. As a consequence, there are choices of the triangle space for which the polyhedron $P=\{\Hat x\in \R^{\Hat N},\,\Hat D \Hat x= \Hat r, 0\leq \Hat x\leq 1\}$ may have not only integral vertices, or more precisely one cannot guarantee this property thanks to total unimodularity. This is summarized in the following theorem.

\begin{theorem}
The incidence matrix associated with any triangle space where each triangle has a large enough number of adjacent neighbors is not totally unimodular. 
\end{theorem}
\begin{proof}
We show in Figure~\ref{fig:counterex} a configuration and, in Table~\ref{tab},  an associated square submatrix of the incidence matrix. The sum of entries over  each line and the sum over each column are even, though the total sum of the matrix entries is not divisible by four. By a result of Camion~\cite{camion-65}, the incidence matrix is not totally unimodular which yields the conclusion according to~\cite{Schrijver-book}[Thm 19.3]. Clearly, any triangle space for which this configuration can occur is also associated to an incidence matrix that is not totally unimodular. 
\end{proof}
It is worth noticing that the previous theorem does not imply that the extremal points of the polyhedron $P$ are necessarily not all integral. It only states that this cannot be guaranteed as usual by the criterion of total unimodularity. We will discuss in the next section what additional informations about integrality can be obtained from a few experiments that we have done using classical solvers for addressing the relaxed linear problem.

\begin{figure}
\begin{center}
\includegraphics[width=0.55 \textwidth,angle=270]{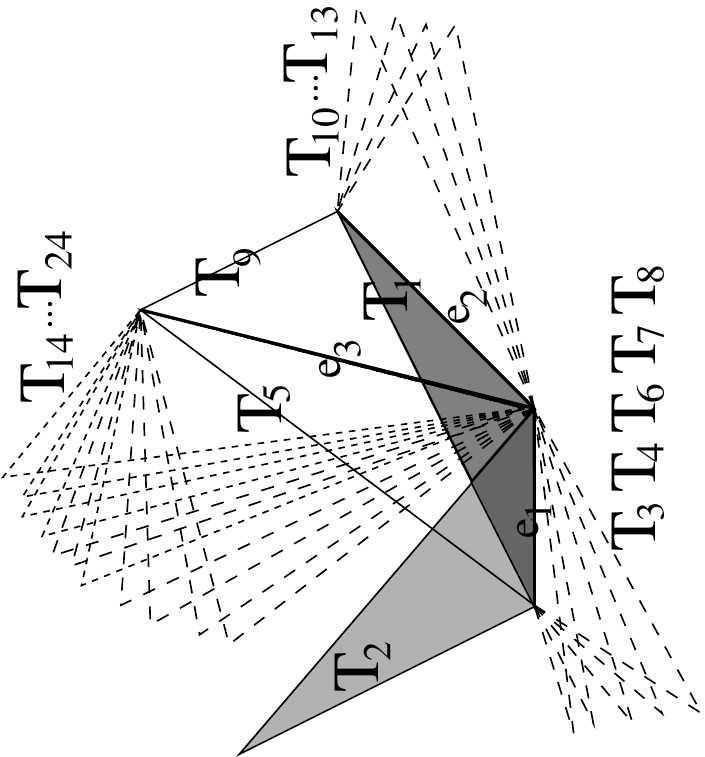}
\caption{A configuration in a triangle space with sufficient resolution. The associated incidence matrix is Eulerian (see text) but does not satisfy Camion's criterion, thus is not totally unimodular.}
\label{fig:counterex}
\end{center}
\end{figure}

\begin{landscape}
\begin{table}
\caption{A square incidence matrix associated with the configuration in Figure~\ref{fig:counterex}. It is Eulerian, i.e. the sum along each line and the sum along each column are even, but the total sum is not divisible by four. According to Camion~\cite{camion-65}, the matrix is not totally unimodular.}\label{tab}
\begin{footnotesize}
\begin{tabular}{@{\hspace{0.2mm}}c@{\hspace{0.2mm}}|@{\hspace{0.2mm}}c@{\hspace{0.2mm}}|@{\hspace{0.2mm}}c@{\hspace{0.2mm}}|@{\hspace{0.2mm}}c@{\hspace{0.2mm}}|@{\hspace{0.2mm}}c@{\hspace{0.2mm}}|@{\hspace{0.2mm}}c@{\hspace{0.2mm}}|@{\hspace{0.2mm}}c@{\hspace{0.2mm}}|@{\hspace{0.2mm}}c@{\hspace{0.2mm}}|@{\hspace{0.2mm}}c@{\hspace{0.2mm}}|@{\hspace{0.2mm}}c@{\hspace{0.2mm}}|@{\hspace{0.2mm}}c@{\hspace{0.2mm}}|@{\hspace{0.2mm}}c@{\hspace{0.2mm}}|@{\hspace{0.2mm}}c@{\hspace{0.2mm}}|@{\hspace{0.2mm}}c@{\hspace{0.2mm}}|@{\hspace{0.2mm}}c@{\hspace{0.2mm}}|@{\hspace{0.2mm}}c@{\hspace{0.2mm}}|@{\hspace{0.2mm}}c@{\hspace{0.2mm}}|@{\hspace{0.2mm}}c@{\hspace{0.2mm}}|@{\hspace{0.2mm}}c@{\hspace{0.2mm}}|@{\hspace{0.2mm}}c@{\hspace{0.2mm}}|@{\hspace{0.2mm}}c@{\hspace{0.2mm}}|@{\hspace{0.2mm}}c@{\hspace{0.2mm}}|@{\hspace{0.2mm}}c@{\hspace{0.2mm}}|@{\hspace{0.2mm}}c@{\hspace{0.2mm}}|@{\hspace{0.2mm}}c@{\hspace{0.2mm}}|@{\hspace{0.2mm}}c@{\hspace{0.2mm}}|@{\hspace{0.2mm}}c@{\hspace{0.2mm}}|@{\hspace{0.2mm}}c@{\hspace{0.2mm}}|@{\hspace{0.2mm}}c@{\hspace{0.2mm}}}
& $T_1$&$T_5$&$T_9$&$T_{1,2}$&$T_{2,3}$&$T_{3,4}$&$T_{4,5}$&$T_{2,5}$&$T_{5,6}$&$T_{1,6}$&$T_{1,7}$&$T_{7,8}$&$T_{2,8}$&$T_{1,9}$&$T_{9,10}$&$T_{10,11}$&$T_{1,11}$&$T_{1,12}$&$T_{12,13}$&$T_{9,13}$&$T_{9,5}$&$T_{5,14}$&$T_{14,15}$&$T_{9,15}$&$T_{9,16}$&$T_{16,17}$&$T_{5,17}$&$\sum$\\
\hline
$(T_1,e_1)$&1&&&-1&&&&&&-1&-1&&&&&&&&&&&&&&&&&-2\\
\hline
$(T_2,e_1)$&&&&-1&-1&&&-1&&&&&-1&&&&&&&&&&&&&&&-4\\
\hline
$(T_3,e_1)$&&&&&-1&-1&&&&&&&&&&&&&&&&&&&&&&-2\\
\hline
$(T_4,e_1)$&&&&&&-1&-1&&&&&&&&&&&&&&&&&&&&&-2\\
\hline
$(T_5,e_1)$&&1&&&&&-1&-1&-1&&&&&&&&&&&&&&&&&&&-2\\
\hline
$(T_6,e_1)$&&&&&&&&&-1&-1&&&&&&&&&&&&&&&&&&-2\\
\hline
$(T_7,e_1)$&&&&&&&&&&&-1&-1&&&&&&&&&&&&&&&&-2\\
\hline
$(T_8,e_1)$&&&&&&&&&&&&-1&-1&&&&&&&&&&&&&&&-2\\
\hline
\hline
$(T_1,e_2)$&1&&&&&&&&&&&&&-1&&&-1&-1&&&&&&&&&&-2\\
\hline
$(T_9,e_2)$&&&1&&&&&&&&&&&-1&-1&&&&&-1&&&&&&&&-2\\
\hline
$(T_{10},e_2)$&&&&&&&&&&&&&&&-1&-1&&&&&&&&&&&&-2\\
\hline
$(T_{11},e_2)$&&&&&&&&&&&&&&&&-1&-1&&&&&&&&&&&-2\\
\hline
$(T_{12},e_2)$&&&&&&&&&&&&&&&&&&-1&-1&&&&&&&&&-2\\
\hline
$(T_{13},e_2)$&&&&&&&&&&&&&&&&&&&-1&-1&&&&&&&&-2\\
\hline
\hline
$(T_9,e_3)$&&&1&&&&&&&&&&&&&&&&&&-1&&&-1&-1&&&-2\\
\hline
$(T_5,e_3)$&&1&&&&&&&&&&&&&&&&&&&-1&-1&&&&&-1&-2\\
\hline
$(T_{14},e_3)$&&&&&&&&&&&&&&&&&&&&&&-1&-1&&&&&-2\\
\hline
$(T_{15},e_3)$&&&&&&&&&&&&&&&&&&&&&&&-1&-1&&&&-2\\
\hline
$(T_{16},e_3)$&&&&&&&&&&&&&&&&&&&&&&&&&-1&-1&&-2\\
\hline
$(T_{17},e_3)$&&&&&&&&&&&&&&&&&&&&&&&&&&-1&-1&-2\\
\hline
\hline
&\multicolumn{26}{c}{plus 7 lines $(T_{18},e_3),\cdots,(T_{24},e_3)$ with only $0$ entries to have a square matrix}&\\
\hline
\hline
$\sum$&2&2&2&-2&-2&-2&-2&-2&-2&-2&-2&-2&-2&-2&-2&-2&-2&-2&-2&-2&-2&-2&-2&-2&-2&-2&-2&-42
\end{tabular}
\end{footnotesize}
\vspace*{0.5cm}
\end{table}
\end{landscape}

\subsection{Testing the relaxed linear problem}
We have tested the relaxed formulation on a few examples at low-resolution using the dual simplex method implemented in the CLP solver. The main reason for using low-resolution is that the number of triangles becomes significantly important as the resolution increases, and both the computational cost and the memory requirements tend to become large. Another reason for working at low-resolution is that there is no need to go high before finding a case of non-integrality. Indeed, consider the examples in figure~\ref{fig:example1}: integral solutions are obtained when the resolution is very low (i.e. when there is no risk to have configurations like in  figure~\ref{fig:counterex}). In the last configuration, however, the optimal solution of the relaxed problem has fractional entries. This confirms that our initial problem cannot be addressed though the classical techniques of relaxation, and with usual LP solvers.

\begin{figure}
\begin{center}
\begin{tabular}{ll}
\includegraphics[width=0.35 \textwidth]{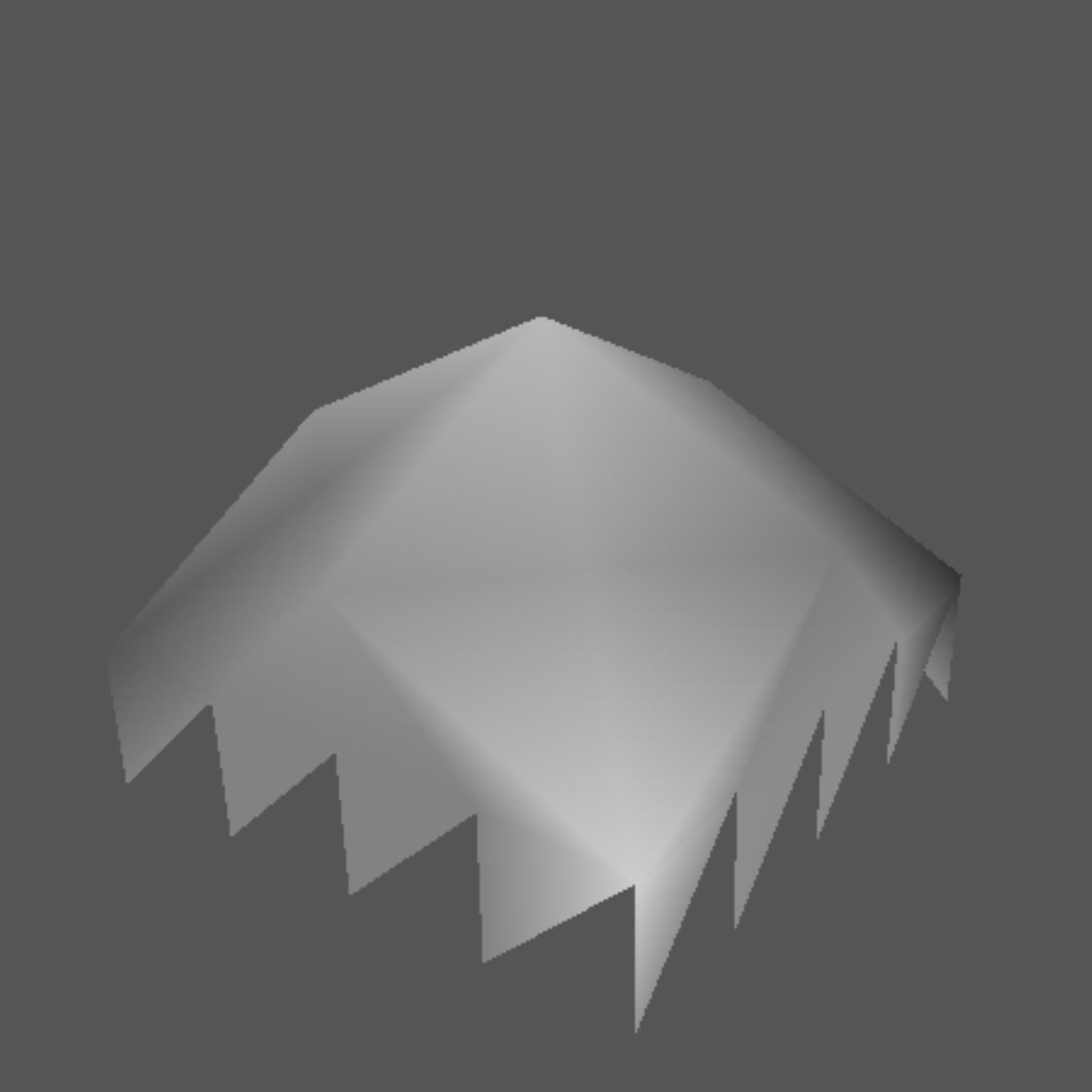}&\includegraphics[width=0.35 \textwidth]{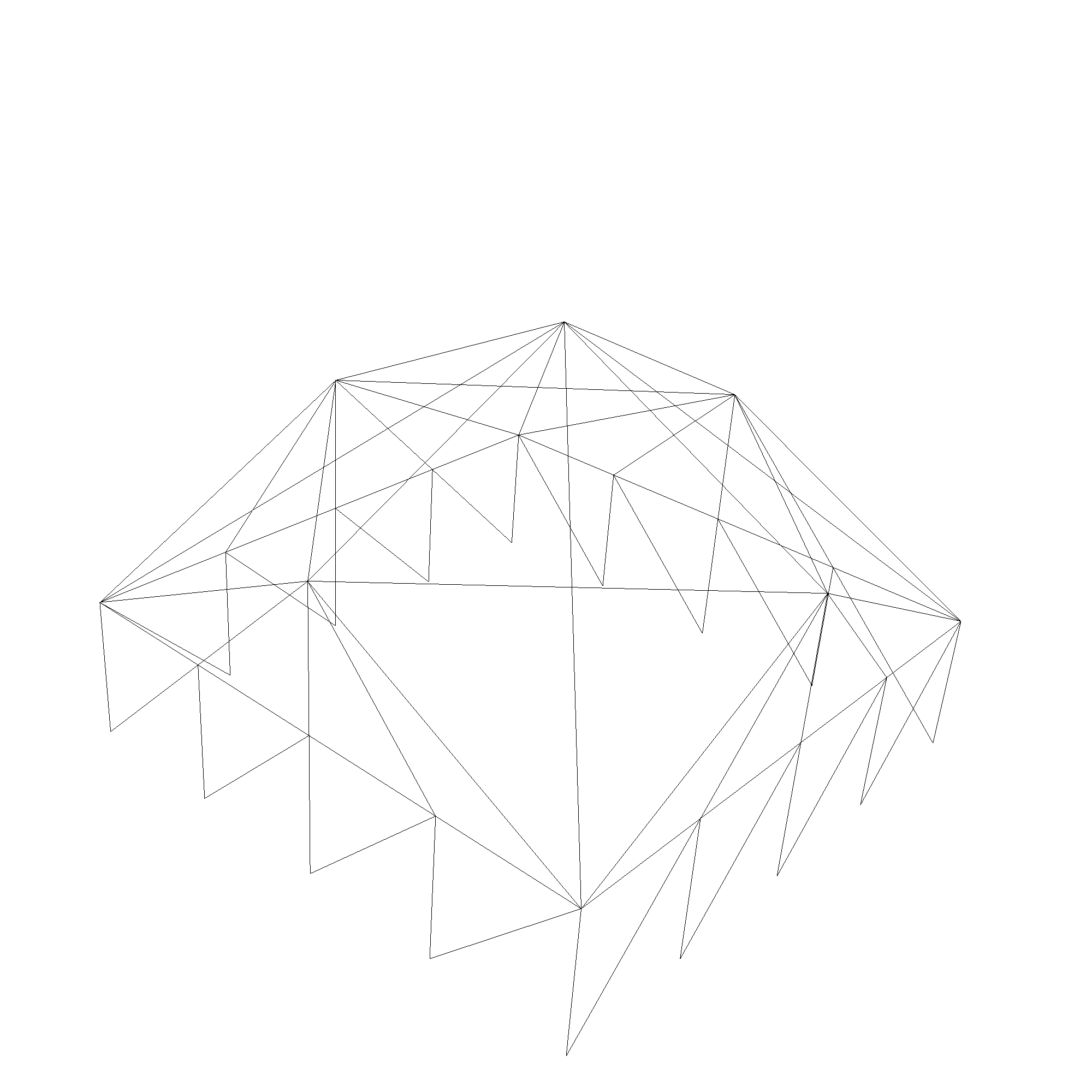}\\
\includegraphics[width=0.35 \textwidth]{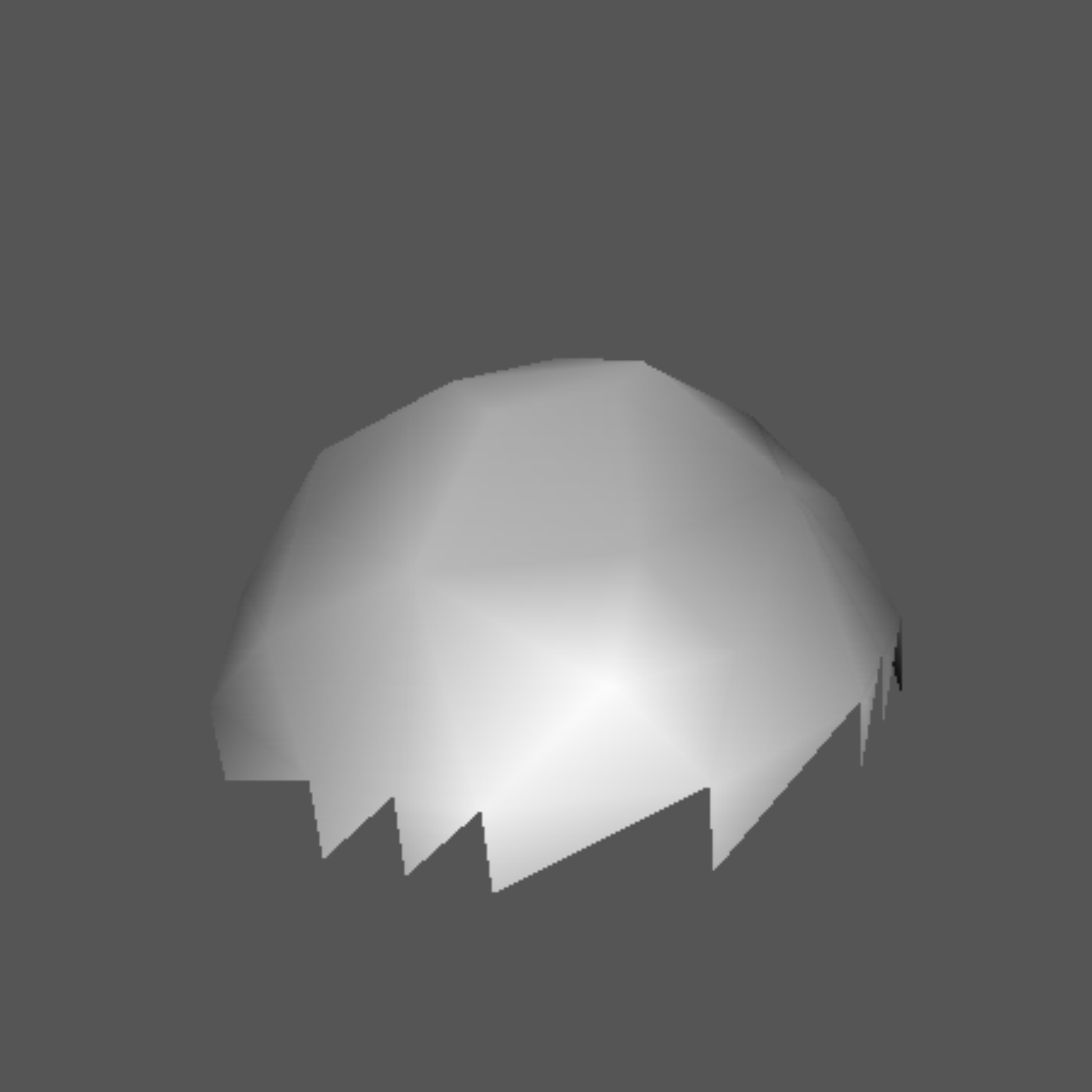}&\includegraphics[width=0.35 \textwidth]{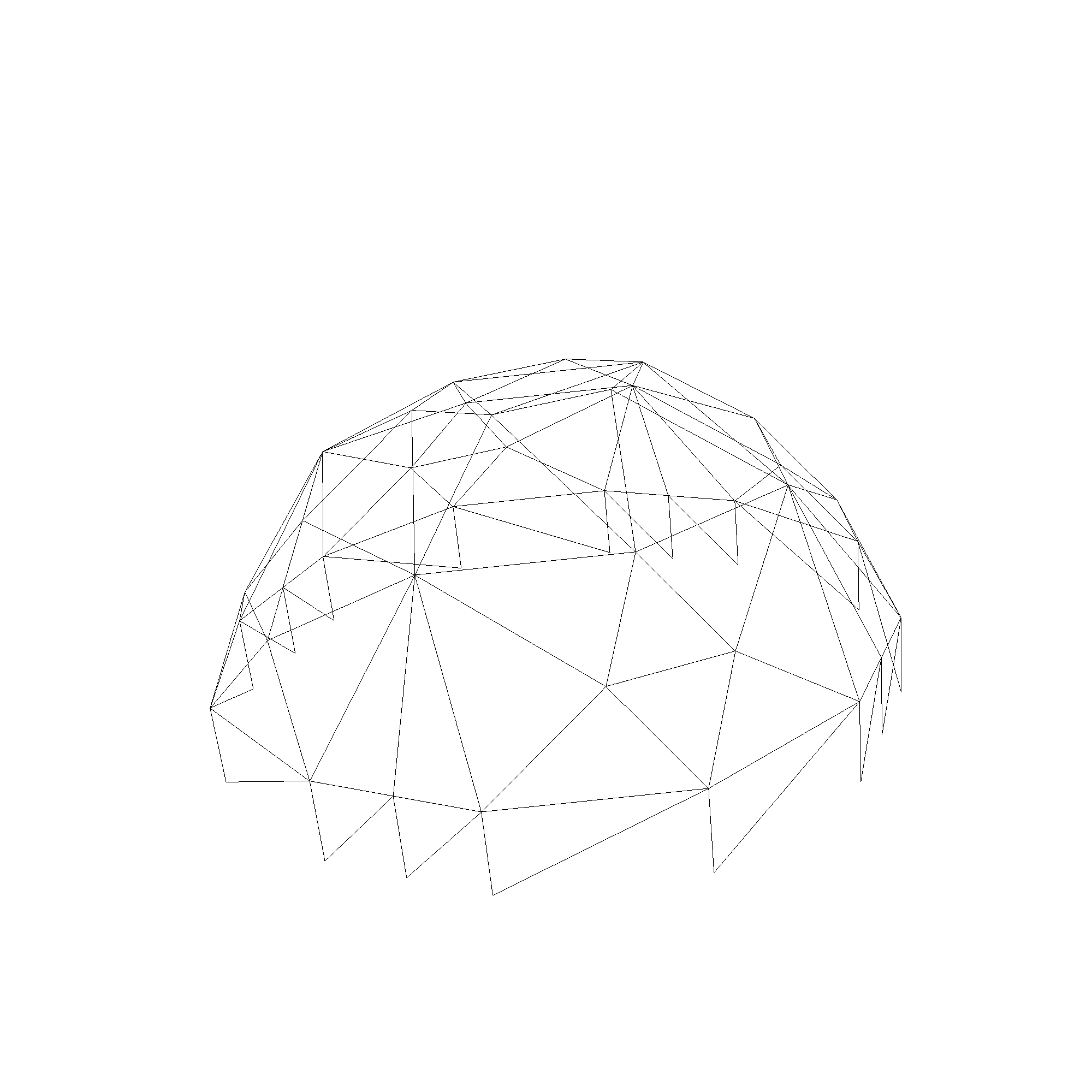}\\
\includegraphics[width=0.35 \textwidth]{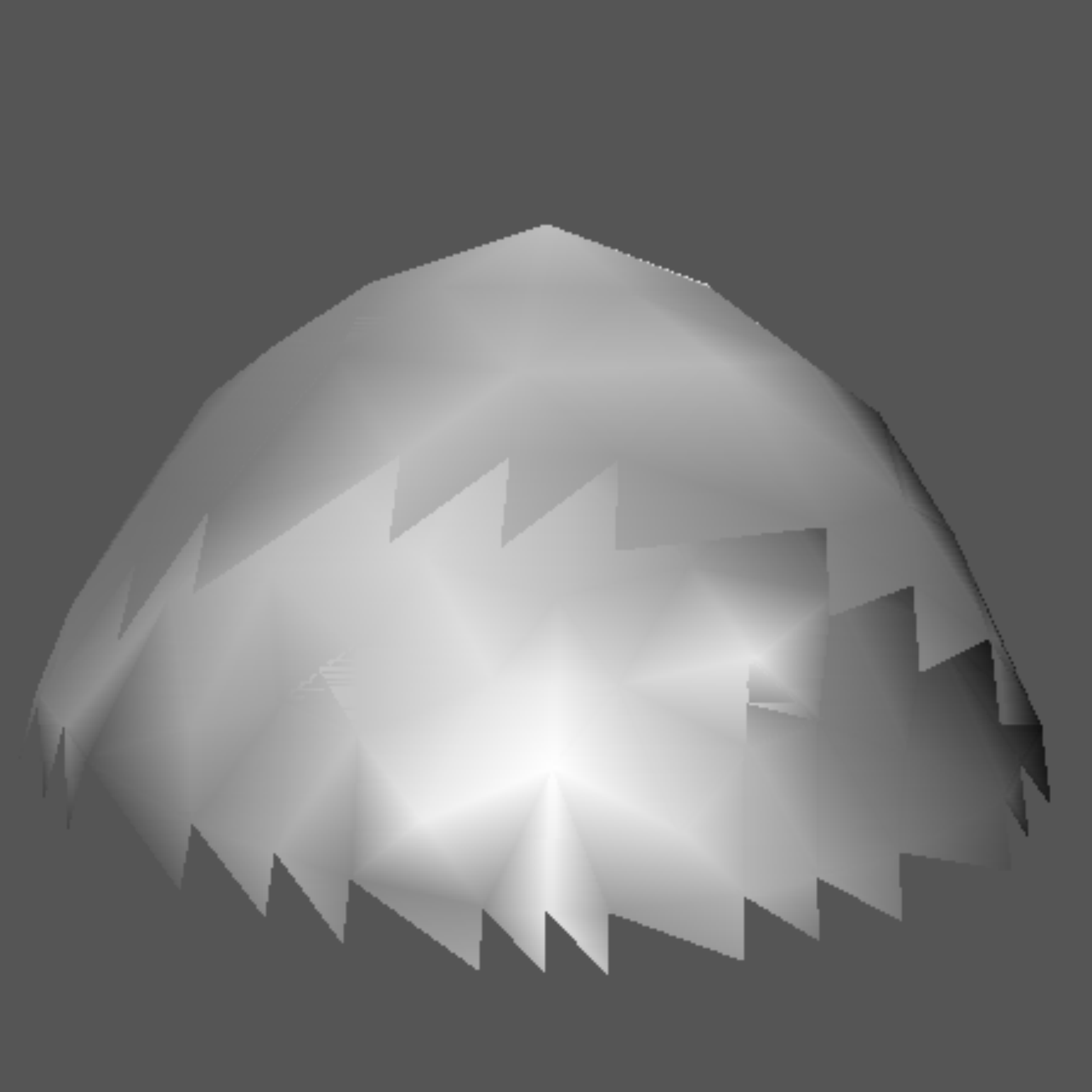}&\includegraphics[width=0.35 \textwidth]{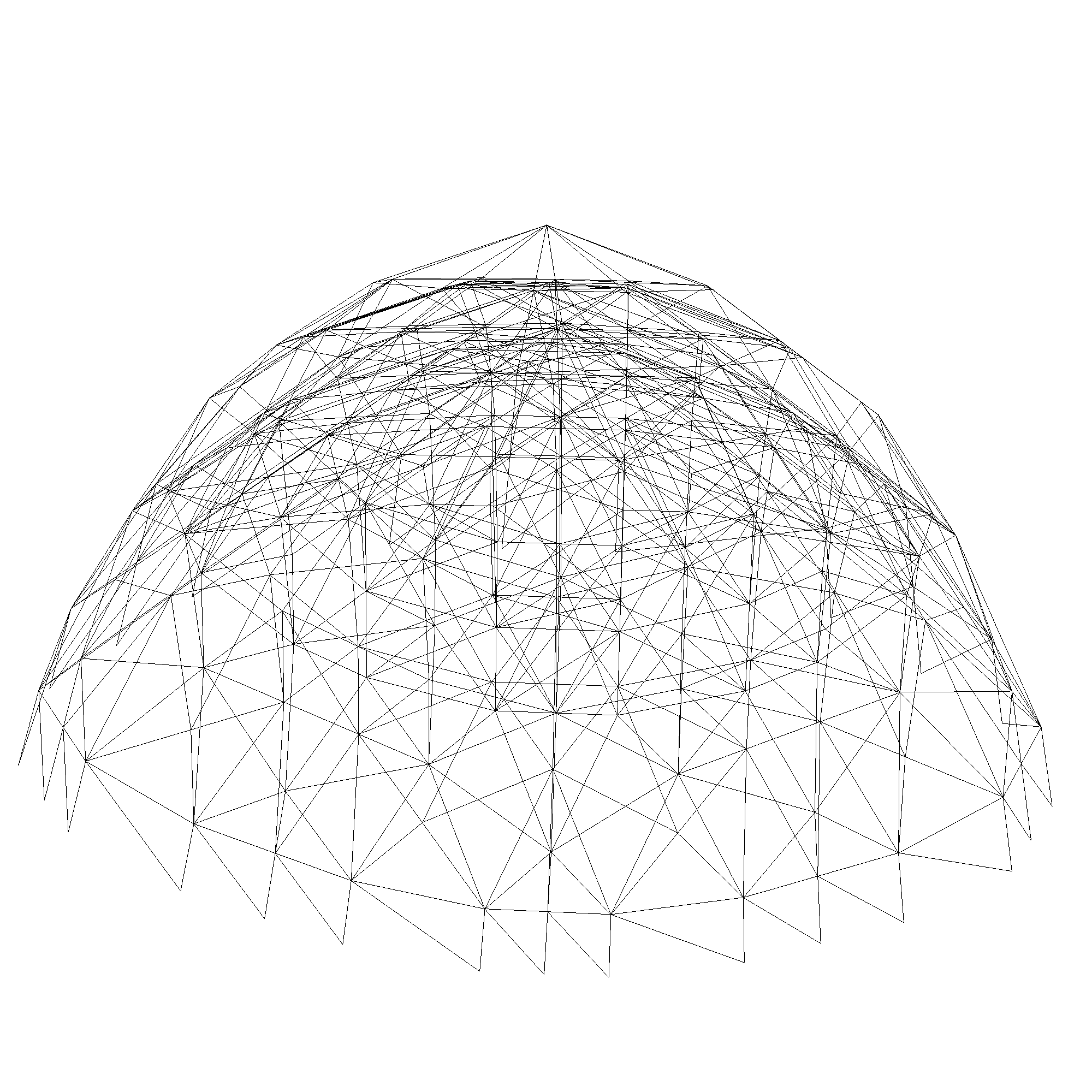}
\end{tabular}
\caption{A series of experiments (the result and the mesh edges) with increasing resolution of the triangle space (and various boundary constraints). An integral solution of the relaxed problem is obtained by a standard LP-solver in both top cases. As for the last case, the triangle space resolution is now large enough for having configurations similar to the counterexample of figure~\ref{fig:counterex}. And indeed, an optimal solution is found for the relaxed problem that is not integral. 
The mesh on the bottom-right shows actually two nested semi-spheres whose triangles have, at least for a few of them, non binary labels.}
\label{fig:example1}
\end{center}
\end{figure}
\subsection{On integer linear programming}
Our results above indicate that, necessarily, integer linear solvers~\cite{Schrijver-book,achterberg-07} should be used. These commonly start with solving the linear programming relaxations, then derive further valid inequalities (called \emph{cuts}) and/or apply a branch-and-bound scheme. Due to the small number of fractional values that we have observed in our experiments, it is quite likely that the derivation of a few cuts only would give integral solutions. However, we did not test this so far because of the running times of this approach: in cases where we get fractional solutions the dual simplex method often needs as long as two weeks and up to $12$ GB memory! From experience with other linear programming problems we consider it likely that the interior point methods implemented in commercial solvers will be much faster here (we expect less than a day). At the same time, we expect the memory consumption to be considerably much higher, so the method would most probably be unusable in practice.

We strongly believe that a specific integer linear solver should be developed rather than using general implementations. It is well known that, for a few problems like the knapsack problem~\cite{Schrijver-book}[chapter 24.6], their specific structure gives rise to ad-hoc efficient approaches. Recalling that our incidence matrix is very sparse and well structured (the nonzero entries of each column are either exactly two $(-1)$, or exactly three $1$) we strongly believe that an efficient integer solver can be developed and our approach can be amenable to higher-resolution results in the near future.

\section{Conclusion}
We have shown that the minimization under boundary constraints of mean curvature based energies over surfaces, and in particular the Willmore energy, can be cast as an integer linear program. Unfortunately, this integer program is not equivalent to its relaxation so the classical LP algorithms offer no warranty that the integer optimal solution will be found.  This implies that pure integer linear algorithms must be used, which are in general much more involved. We believe however that the particular structure of the problem paves the way to a dedicated algorithm that would provide high-resolution {\it global} minimizers of the Willmore boundary problem and generalizations. This is the purpose of future research.
\bibliographystyle{splncs} 

\end{document}